\documentclass[11pt]{article}
\usepackage{amssymb}
\usepackage{colortbl}
\usepackage{amsfonts,amsmath, longtable}

\usepackage{comment}

        \def\theequation{\thesection.\arabic{equation}}

\newcommand{\mF}{{\mathcal F}}
\newcommand{\mR}{{\mathcal R}}
\newcommand{\mH}{{\mathcal H}}

\newcommand{\mD}{{\mathcal D}}

\newcommand{\al}{\alpha}

\newcommand{\vth}{\vartheta}

\newcommand{\mC}{\mathbb C}
\newcommand{\mZ}{\mathbb Z}

\newenvironment{proof}{\par\noindent{\bf Proof.}}{\hfill$\scriptstyle\blacksquare$}

\def\beq{\begin{equation}}
\def\eq{\end{equation}}
\def\p{\partial}

\newtheorem{theor}{Theorem}

\def\res{\mathop{\hbox{Res}}\limits}

\begin{document}

\setcounter{page}{1}

\begin{center}

\

\vspace{-0mm}

{\Large{\bf  Supersymmetric generalization of q-deformed}}

\vspace{3mm}

{\Large{\bf long-range  spin chains  of Haldane-Shastry type and trigonometric}}

 \vspace{3mm}

{\Large{\bf ${\rm GL}(N|M)$ solution of associative Yang-Baxter equation}}

 \vspace{15mm}

 {\Large {M. Matushko}}
\qquad\quad\quad
 {\Large {A. Zotov}}

  \vspace{7mm}

{\em Steklov Mathematical Institute of Russian
Academy of Sciences,\\ Gubkina str. 8, 119991, Moscow, Russia}



   \vspace{5mm}

 {\small\rm {e-mails: matushko@mi-ras.ru, zotov@mi-ras.ru}}

\end{center}

\vspace{0mm}

\begin{abstract}
We propose commuting sets of matrix-valued difference operators in terms of trigonometric  ${\rm  GL}(N|M)$-valued $R$-matrices thus providing quantum supersymmetric (and possibly anisotropic) spin Ruijsenaars-Macdonald operators.
 Two types of trigonometric supersymmetric $R$-matrices are used for this purpose. The first is the one related to the affine quantized algebra ${\hat{\mathcal U}}_q({\rm  gl}(N|M))$. The second is a graded version of the standard
 $\mZ_n$-invariant $A_{n-1}$ type $R$-matrix. We show that being properly normalized the latter graded $R$-matrix satisfies the associative Yang-Baxter equation. Next, we discuss construction of long-range spin chains
using the Polychronakos freezing trick.
As a result we obtain a new family of spin chains, which extends the ${\rm  gl}(N|M)$-invariant Haldane-Shastry spin chain to q-deformed case with possible presence of anisotropy.
\end{abstract}

\newpage
{\small{ \tableofcontents }}



\section{Introduction}\label{sec1}
\setcounter{equation}{0}

The Haldane-Shastry long-range spin chain \cite{HS1} is a quantum integrable model describing
interactions between $L$ spins distributed equidistantly on a circle. The Hamiltonian is
 \beq\label{HS}
 \begin{array}{c}
  \displaystyle{
 H^{\rm{HS}}=
  \frac{\pi^2}{2}\sum\limits_{i\neq j}^L \frac{1-P_{ij}}{\sin^2(\pi(x_i-x_j))}\,,
  }
 \end{array}
\eq
 where $x_k=k/L$, $k=1,...,L$ and $P_{ij}$ is the permutation operator representing exchange interaction between spins;
 for the ${\rm gl}_n$ model
 it acts on the Hilbert space $\mH=(\mC^n)^{\otimes L}$ by permuting the $i$-th and the $j$-th tensor components.
The Haldane-Shastry model is naturally extended to the supersymmetric case \cite{Haldane,BMB}; the Hamiltonian  (\ref{HS}) keeps the same form but the permutation operator is defined on the graded space $\mH=(\mC^{N|M})^{\otimes L}$ (see (\ref{supertrans}) below).

The non-supersymmetric model has several important generalizations.
The first one provides anisotropic analogues \cite{SeZ,MZ2} of (\ref{HS}) likewise it happens in the short-range spin chains (magnets), where the isotropic XXX model is extended to the anisotropic XXZ and XYZ models. For example, the XXZ version of the Haldane-Shastry model (\ref{HS}) for the ${\rm gl}_2$ case  is
 \beq\label{s0911}
 \begin{array}{c}
  \displaystyle{
 H^{\rm{XXZ}}=
  -\frac{\pi^2}{2}\sum\limits_{i\neq j}^L
  \frac{\cos(\pi(x_i-x_j))(\sigma_1^{(i)}\sigma_1^{(j)}+
  \sigma_2^{(i)}\sigma_2^{(j)})+\sigma_3^{(i)}\sigma_3^{(j)}
   }{\sin^2(\pi(x_i-x_j))} \,,
  }
 \end{array}
\eq
 where $\sigma_\al^{(i)}$ denotes the Pauli $\sigma_\al$ in the $i$-th tensor component.

 The second type of generalization is a q-deformation, which is similar to the relativistic generalization
 of the Calogero-Moser-Sutherland differential operators to the Ruijsenaars-Schneider or Macdonald difference operators.
 For the ${\rm gl}_2$ Haldane-Shastry model (\ref{HS}) the q-deformation was suggested in \cite{Uglov}. Then it was studied
 in \cite{Lam1,Lam2}, where the Hamiltonian was presented in the following form ($q=e^{\pi i\hbar}$):
  \beq\label{s461re0}
\begin{array}{c}
\displaystyle{
 H^{q\rm HS}=-\pi\sum\limits_{k<i}^L \frac{\sin(\pi \hbar)}{\sin \pi(\hbar+x_i-x_k)\sin\pi(\hbar-x_i+x_k)}\, \bar{\mR}_{i-1,i}\dots \bar{\mR}_{k+1,i}C_{k,i}\bar{\mR}_{i,k+1}\dots \bar{\mR}_{i,i-1}
}
\end{array}
\eq
with
\beq\label{s461re01}
  C_{12}=\displaystyle{
 \left(e^{-\imath\pi \hbar} e_{11}\otimes e_{22}-e_{12}\otimes e_{21}-e_{21}\otimes e_{12}+e^{\imath\pi \hbar}e_{22}\otimes e_{11}\right)}
  \eq
  and $\bar{\mR}_{i,j}=\bar\mR^\hbar(x_i-x_j)$ is the quantum $R$-matrix for the affine quantized algebra ${\hat{\mathcal U}}_q({\rm  gl}_2)$,
  see (\ref{2,0n}).
  In (\ref{s461re01}) $e_{ij}$ stands for the standard matrix units.
  This result was extended in \cite{MZ1,MZ2} to a wide class of ${\rm GL}_n$ $R$-matrices including the elliptic Baxter-Belavin
  $R$-matrix and its trigonometric degenerations. The idea was first to define a set of matrix-valued commuting difference operators $\mD_k$, $k=1,...,L$ generalizing the (scalar) commuting difference Ruijsenaars-Macdonald operators \cite{Ruij,Macd}. The commutativity property $[\mD_i,\mD_j]=0$
  was shown to be equivalent to a set of $R$-matrix identities, which was then proved by analytical methods
  using explicit form of $R$-matrices in the fundamental representation of ${\rm GL}_n$. Next, using the Polychronakos
  freezing trick \cite{Polych,Lam3}, the commuting Hamiltonains of q-deformed Haldane-Shastry type models were found.
  A similar idea will be used in this paper for the case of supersymmetric $R$-matrices.
There are extensions of the Haldane-Shastry model to different root systems \cite{BFG} although we do not discuss them in this paper.

{\bf The goal of the paper} is three-fold. First, we describe two types of supersymmetric trigonometric
$R$-matrices which we deal with. One is the widely known $R$-matrix related to the affine quantized algebra ${\hat{\mathcal U}}_q({\rm  gl}(N|M))$ \cite{BS}. Another one is a supersymmetrization of $\mZ_n$-invariant $A_{n-1}$ type $R$-matrix \cite{Chered2}. We suggest an answer for it and could not find it in the literature although
we believe it is known, since the graded
extensions of trigonometric $R$-matrices were extensively studied \cite{KS}. In order to relate our $R$-matrix to known results we construct a Drinfeld-type twist \cite{Drinfeld}, which transforms it into the ${\hat{\mathcal U}}_q({\rm  gl}(N|M))$ case. For the graded $\mZ_n$-invariant $R$-matrix we also show that (after being properly normalized) it satisfies
the so-called associative Yang-Baxter equation \cite{FK}. This result is obtained as a by-product. At the same time it is important by itself since it extends the classification of trigonometric solutions of the associative Yang-Baxter equation \cite{SchP} to the graded case. Also, solutions of the associative Yang-Baxter equation are used for
different constructions in integrable systems \cite{FK,LOZ14,KrZ,GZ,SeZ,SeZ2} including the above mentioned $R$-matrix identities \cite{MZ1}.

Second, we prove that the main statement of \cite{MZ1} concerning construction
of commuting matrix-valued difference operators is valid for the graded $R$-matrices under consideration.
In this way we obtain two sets (corresponding to two possible choices of $R$-matrices) of commuting graded matrix-valued difference operators. The operators can be viewed as Hamiltonians of the graded (or supersymmetric) spin Ruijsenaars-Schneider model or its anisotropic versions. The spin Ruijsenaars-Schneider model was introduced at
the classical level in \cite{KrichZ} and much progress was achieved in studies of corresponding Poisson structure and its
quantization \cite{ACF}. The supersymmetric version of the quantum trigonometric Ruijsenaars-Schneider model
was introduced in \cite{BDM}. The quantization in \cite{BDM} was based on the usage of Cherednik operators. Our approach is different since our difference operators commute not for any $R$-matrix but only for those which
satisfy a certain set of identities. It is also interesting to compare our results with the classical construction
of the spin Ruijsenaars-Schneider model from \cite{KL}.

Finally, we use the Polychronakos
  freezing trick \cite{Polych,Lam3}
   and obtain supersymmetric generalizations of q-deformed Haldane-Shastry type models.
The first Hamiltonian has the form:
\beq\label{s4611}
\begin{array}{c}
\displaystyle{
H_1 =\sum\limits_{k<i}^L \bar{R}_{i-1,i}\dots \bar{R}_{k+1,i}\bar{R}_{k,i}\bar{F}_{i,k}\bar{R}_{i,k+1}\dots \bar{R}_{i,i-1}\,,
}
\end{array}
\eq
where
${\bar F}^\hbar_{ij}(z)={\partial_z}\bar{R}^\hbar_{ij}(z)$, ${\bar F}_{ij}={\bar F}^\hbar_{ij}(x_i-x_j)$ and
${\bar R}_{ij}={\bar R}^\hbar_{ij}(x_i-x_j)$. It is valid for both graded $R$-matrices which we consider, i.e.
for (\ref{uqsupern}) and (\ref{superRn}).
For example, in the supersymmetric case ${\hat{\mathcal U}}_q({\rm  gl}(1|1))$ related to the $R$-matrix of (\ref{1,1n})
one gets the same expression as in (\ref{s461re0}) but with $C_{12}$ replaced by
\beq\label{Csusy}
  C^{\rm susy}_{12}=
  \displaystyle{
 \left(e^{-\imath\pi \hbar} e_{11}\otimes e_{22}+e_{12}\otimes e_{21}-e_{21}\otimes e_{12}+e^{\imath\pi \hbar}e_{22}\otimes e_{11}+2\cos(\pi \hbar) e_{22}\otimes e_{22}\right)\,,
 }
  \eq
  where $e_{ij}$ are now matrix units in ${\rm End}(\mC^{1|1})$.
  In the non-relativistic limit ($q\rightarrow 1$ or $\hbar\rightarrow 0$) one reproduces
  the supersymmetric Haldane-Shastry Hamiltonian (\ref{HS}) with graded permutation operator $P_{ij}$.
 At the same time the model related to the $R$-matrix (\ref{g1,1n}) in the non-relativistic limit
 provides the ${\rm gl}(1|1)$ graded version of the model (\ref{s0911}):
 \beq\label{q410}
\begin{array}{c}
\displaystyle{
H^{\rm susy XXZ}=\pi\sum_{k<i}^L
\frac{ e_{11}\otimes e_{22}+e_{22}\otimes e_{11}+2 e_{22}\otimes e_{22}
+\cos (\pi(x_i-x_k))(e_{12}\otimes e_{21}-e_{21}\otimes e_{12})}{\sin^2(\pi(x_i-x_k))}\,.
}
\end{array}
\eq

{\bf The paper is organized as follows.} In the next Section the supersymmetric trigonometric $R$-matrices are described. In Section 3 we formulate the main properties of these $R$-matrices
including the associative Yang-Baxter equation (and the quantum Yang-Baxter equation)
for the graded version of $\mZ_n$-invariant $R$-matrix and
the associative Yang-Baxter equation with additional term for
the $R$-matrix of the affine quantized algebra ${\hat{\mathcal U}}_q({\rm gl}(N|M))$. The
twist transformation relating two $R$-matrices is given as well.
In Section 4 the construction of commuting spin (or matrix) difference operators is given.
Finally, in Section 5 we apply the freezing trick to obtain new families of integrable
long-range spin chains based on supersymmetric $R$-matrices. In Appendix A we give a sketch of the proof that the graded version of $\mZ_n$-invariant $R$-matrix
satisfies the associative Yang-Baxter equation. Appendix B contains
the list of explicit expressions for the normalized graded trigonometric $R$-matrices.

\section{Supersymmetric R-matrices}\label{sec2}
\setcounter{equation}{0}

\subsection{Notations}

In the supersymmetric case consider the $\mZ_2$-graded
vector space $\mC^{N|M}=\mC^{N}\oplus\mC^{M}$ with grading
\beq \label{grading}
p_i=\begin{cases}
     0 \text{     for    } 1\le i\le N, \\
     1 \text{     for   } N< i\le N+M.
    \end{cases}
\eq
Denote by $e_{ij}$ the elementary matrix units acting on the space $\mC^{N|M}$. Multiplication in the tensor product  $\mC^{N|M}\otimes\mC^{N|M}$ is defined by
\beq\label{superrules}
(\mathbb{Id}\otimes e_{ij})(e_{kl}\otimes \mathbb{Id})=(-1)^{(p_i+p_j)(p_k+p_l)}e_{kl}\otimes e_{ij}
\eq
and the graded permutation operator acting on $\mC^{N|M}\otimes\mC^{N|M}$  is given by
\beq\label{supertrans}
P_{12}=\sum_{i,j=1}^{N+M} (-1)^{p_j}e_{ij}\otimes e_{ji}\,.
\eq

\subsection{$R$-matrices}
  We consider two types of trigonometric supersymmetric $R$-matrices. The first is related to \\ the \underline{affine quantized algebra ${\hat{\mathcal U}}_q({\rm gl}(N|M))$ \cite{BS}:}
%
  \beq\label{uqsuper}
   \begin{array}{c}
   \displaystyle{
  \mR^\hbar_{12}(z)=
   \pi\sum\limits_{a=1}^{N+M} \Big((-1)^{p_a}\cot(\pi z)+\coth(\pi\hbar)\Big)e_{aa}\otimes e_{aa}
 +
  }
  \\ \ \\
   \displaystyle{
 +\frac{\pi}{\sin(\pi\hbar)}\sum\limits_{a\neq b}^{N+M} e_{aa}\otimes e_{bb}+
 \frac{\pi}{\sin(\pi z)}\sum\limits_{a< b}^{N+M}
 \Big( (-1)^{p_b} e_{ab}\otimes e_{ba}\,e^{\pi\imath z}+(-1)^{p_a}e_{ba}\otimes
 e_{ab}\,e^{-\pi\imath z}\Big)\,.
  }
  \end{array}
  \eq
  The normalized version of this $R$-matrix is given in (\ref{uqsupern}). Let us give two simple examples.

    \paragraph{Example ${\hat{\mathcal U}}_q({\rm gl}(2|0))$:}
  \beq\label{2,0}
   \begin{array}{c}
   \displaystyle{
  \mR^\hbar_{12}(z)=
   \pi\Big(\cot(\pi z)+\coth(\pi\hbar)\Big)\Big(e_{11}\otimes e_{11}+e_{22}\otimes e_{22}\Big)+
  }
  \\ \ \\
   \displaystyle{
 +\frac{\pi}{\sin(\pi\hbar)}\Big(e_{11}\otimes e_{22}+e_{22}\otimes e_{11}\Big)+
 \frac{\pi}{\sin(\pi z)}\Big(e_{12}\otimes e_{21}e^{\pi\imath z}+e_{21}\otimes e_{12}e^{-\pi\imath z}\Big)\,.
  }
  \end{array}
  \eq
  This $R$-matrix (more precisely, its normalized version (\ref{2,0n})) provides the q-deformed Haldane-Shastry model (\ref{s461re0})-(\ref{s461re01}).
 The (semiclassical) limit $\hbar\rightarrow 0$ to the classical $r$-matrix
 corresponds to non-relativistic
 limit
  $q=e^{\pi i\hbar}\rightarrow 1$, which leads
 to the ordinary Haldane-Shastry model (\ref{HS}).

 Let us remark that the parameter $\hbar$ in a quantum $R$-matrix plays the role of the Planck constant providing
 the semiclassical limit $\hbar\rightarrow 0$: $\displaystyle R_{12}^\hbar(z)=\frac{1}{\hbar}\,1_{N|M}\otimes 1_{N|M}+r_{12}(z)+O(\hbar)$. At the same time
 it plays the role of the relativistic deformation parameter in the sense that the limit $\hbar\rightarrow 0$ corresponds also to
 transition from Lie group to Lie algebra. For example, quantum $R$-matrices can be used at the classical level for constructing ''relativistic'' integrable systems (those described by quadratic Poisson brackets of the Sklyanin type) with relativistic parameter $\hbar$. A direct relation to the Ruijsenaars-Schneider model with the relativistic parameter $\hbar$ can be established, see \cite{LOZ_rel,KrZ}.

   \paragraph{Example ${\hat{\mathcal U}}_q({\rm gl}(1|1))$:}
\beq\label{1,1}
   \begin{array}{c}
   \displaystyle{
  \mR^\hbar_{12}(z)=
   \pi\Big(\cot(\pi z)+\coth(\pi\hbar)\Big)e_{11}\otimes e_{11}
  +\pi\Big(-\cot(\pi z)+\coth(\pi\hbar)\Big)e_{22}\otimes e_{22}+
  }
  \\ \ \\
   \displaystyle{
 +\frac{\pi}{\sin(\pi\hbar)}\Big(e_{11}\otimes e_{22}+e_{22}\otimes e_{11}\Big)+
 \frac{\pi}{\sin(\pi z)}\Big(-e_{12}\otimes e_{21}e^{\pi\imath z}+e_{21}\otimes e_{12}e^{-\pi\imath z}\Big)\,.
  }
  \end{array}
  \eq
The normalized version (\ref{1,1n}) of the latter $R$-matrix provides the supersymmetric version of q-deformed Haldane-Shastry model (\ref{Csusy}).

The next one is the \underline{graded extension of $\mZ_n$-invariant $R$-matrix:}
   \begin{equation}\label{superR}
   \begin{array}{c}
   \displaystyle{
  R^\hbar_{12}(z)
  =\pi\sum\limits_{a=1}^{N+M}\Big((-1)^{p_a}\cot(\pi z)+\cot(\pi\hbar)\Big) e_{aa}\otimes e_{aa}+
   }
  \\ \ \\
     \displaystyle{
  +\pi\sum\limits_{a\neq c}^{N+M} e_{aa}\otimes e_{cc}\frac{ \exp\Big(\frac{\pi\imath\hbar}{N+M}\Big(2(a-c)-(N+M){\rm sign}(a-c)\Big)\Big)}{\sin(\pi\hbar)}+
  }
  \\ \ \\
   \displaystyle{
 +\pi\sum\limits_{a\neq c}^{N+M} (-1)^{p_c} e_{ac}\otimes e_{ca}\,\frac{ \exp\Big(\frac{\pi\imath z}{N+M}\Big(2(a-c)-(N+M){\rm sign}(a-c)\Big)\Big)}{\sin(\pi z)}\,.
  }
  \end{array}
  \end{equation}
  In the non-graded case ($M=0$) it was introduced in \cite{Chered2} and obtained in the trigonometric limit from the elliptic $R$-matrix in \cite{AHZ}. The graded version (\ref{superR}) is presumably also known but we could not
  find exactly this one in the literature. Therefore, the claim that (\ref{superR}) is a quantum $R$-matrix requires a proof that it satisfies the quantum Yang-Baxter equation and this will be proved in the next Section.
  The normalized version of (\ref{superR}) is given in (\ref{superRn}). Let us write down two simplest examples of non-graded and graded cases.

    \paragraph{Example ${\rm GL}(2|0)$:}
  \beq\label{g2,0}
   \begin{array}{c}
   \displaystyle{
  R^\hbar_{12}(z)=
   \pi\Big(\cot(\pi z)+\coth(\pi\hbar)\Big)\Big(e_{11}\otimes e_{11}+e_{22}\otimes e_{22}\Big)+
  }
  \\ \ \\
   \displaystyle{
 +\frac{\pi}{\sin(\pi\hbar)}\Big(e_{11}\otimes e_{22}+e_{22}\otimes e_{11}\Big)+
 \frac{\pi}{\sin(\pi z)}\Big(e_{12}\otimes e_{21}+e_{21}\otimes e_{12}\Big)\,.
  }
  \end{array}
  \eq
 This $R$-matrix (more precisely, its normalized version (\ref{g2,0n})) provides a q-deformed version of the XXZ anisotropic model (\ref{s0911}).
 Construction of the Hamiltonians will be reviewed in Section \ref{sec5}, see (\ref{s431}) and (\ref{s461})-(\ref{s47}).

    \paragraph{Example ${\rm GL}(1|1)$:}
  \beq\label{g1,1}
   \begin{array}{c}
   \displaystyle{
  R^\hbar_{12}(z)=
   \pi\Big(\cot(\pi z)+\coth(\pi\hbar)\Big)e_{11}\otimes e_{11}
  +\pi\Big(-\cot(\pi z)+\coth(\pi\hbar)\Big)e_{22}\otimes e_{22}+
  }
  \\ \ \\
   \displaystyle{
 +\frac{\pi}{\sin(\pi\hbar)}\Big(e_{11}\otimes e_{22}+e_{22}\otimes e_{11}\Big)+
 \frac{\pi}{\sin(\pi z)}\Big(-e_{12}\otimes e_{21}+e_{21}\otimes e_{12}\Big)\,.
  }
  \end{array}
  \eq
The normalized version of the latter $R$-matrix (\ref{g1,1n})) leads to the graded version of the q-deformed XXZ (anisotropic) model.

\section{Yang-Baxter equations and other properties of $R$-matrices}\label{sec4}
\setcounter{equation}{0}
Any quantum $R$-matrix, by definition,  satisfies the quantum Yang-Baxter equation:
\beq\label{QYB}
\begin{array}{c}
\displaystyle{
    R^{\hbar}_{12}(u)  R^{\hbar}_{13}(u+v) R^{\hbar}_{23}(v) =
      R^{\hbar}_{23}(v) R^{\hbar}_{13}(u+v) R^{\hbar}_{12}(u)\,.
      }
\end{array}\eq
For the $R$-matrix (\ref{uqsuper}) this property is known. Below we prove it for (\ref{superR}).

\subsection{Associative Yang-Baxter equation}
In this subsection we discuss the associative Yang-Baxter equation \cite{FK}:
\beq\label{AYBE}
\begin{array}{c}
    R^{x}_{12}(z_1-z_2) R^{y}_{23}(z_2-z_3) = R^{y}_{13}(z_1-z_3) R^{x-y}_{12}(z_1-z_2) + R^{y-x}_{23}(z_2-z_3) R^{x}_{13}(z_1-z_3)\,.
\end{array}
\eq
The Yang-Baxter equations (\ref{QYB}) and (\ref{AYBE}) have different sets of solutions.
This is easy to see from the scalar (i.e. ${\rm gl}_1$) case. Indeed, the quantum Yang-Baxter equation
(\ref{QYB}) becomes an identity in this case, while (\ref{AYBE}) is a nontrivial
functional equation. It has solution given by the following function\footnote{This is a trigonometric limit of the elliptic Kronecker function $\phi(z)=\frac{\vth'(0)\vth(\hbar+z)}{\vth(z)\vth(\hbar)}$.}:
 \beq\label{a0811}
 \begin{array}{c}
  \displaystyle{
\phi(x,z)=\pi\cot(\pi x)+\pi\cot(\pi z)=\frac{\pi\sin(\pi(x+z))}{\sin(\pi x)\sin(\pi z)}\,.
 }
 \end{array}
 \eq
However, there is a class of $R$-matrices satisfying both Yang-Baxter equations (\ref{QYB}) and (\ref{AYBE}).
It includes the elliptic ${\rm GL}_n$-valued $R$-matrix (in the fundamental representation of ${\rm GL}_n$) and its degenerations. Trigonometric ${\rm GL}_n$ solutions
of (\ref{AYBE}) were classified in \cite{SchP}. The non-graded $R$-matrix (\ref{superR}) with $M=0$ is contained
in this classification (see \cite{KrZ}) as well as more general (the so-called non-standard) trigonometric $R$-matrix found previously in \cite{AHZ}.

As for the $R$-matrix related to ${\hat{\mathcal U}}_q({\rm gl}_n)$, it was mentioned in \cite{KrZ} that
the non-graded $R$-matrix (\ref{uqsuper}) with $M=0$ satisfies the associative Yang-Baxter equation
with additional term:
\beq\label{relAY0}
\begin{array}{c}
\displaystyle{
    R^{x}_{12}(z_1-z_2) R^{y}_{23}(z_2-z_3) - R^{y}_{13}(z_1-z_3) R^{x-y}_{12}(z_1-z_2) - R^{y-x}_{23}(z_2-z_3) R^{x}_{13}(z_1-z_3)=
    }
    \\ \ \\
    \displaystyle{
    =\frac{\pi}{2\cos(\frac{\pi x}{2})\cos(\frac{\pi y}{2})\cos(\frac{\pi (x-y)}{2})}\sum_{a\neq b\neq c\neq a}^{n} e_{aa}\otimes e_{bb}\otimes e_{cc}\,.
    }
\end{array}\eq
The r.h.s. of (\ref{relAY0}) vanishes for $n\leq 2$ so that (\ref{2,0}) satisfies (\ref{AYBE}).
But for $n>2$ the expression on the r.h.s. of (\ref{relAY0}) is nontrivial although  it is independent of
spectral parameters $z_1,z_2,z_3$.

Similar results are valid for supersymmetric $R$-matrices (\ref{uqsuper}) and (\ref{superR}). Notice that
both $R$-matrices (\ref{uqsuper}) and (\ref{superR}) turn into the function $\phi(\hbar,z)$ (\ref{a0811})
when $N+M=1$, so that (\ref{AYBE}) is satisfied in this case.

\begin{theor}\label{th1}
The supersymmetric extension of the $\mZ_n$-invariant $R$-matrix (\ref{superR}) satisfies the associative Yang-Baxter equation (\ref{AYBE}).
\end{theor}

\begin{theor}\label{th2}
The supersymmetric $R$-matrix (\ref{uqsuper}) of the affine quantized algebra ${\hat{\mathcal U}}_q({\rm gl}(N|M))$ satisfies the associative Yang-Baxter equation (\ref{AYBE}) with additional term:
\beq\label{relAY}
\begin{array}{c}
\displaystyle{
    R^{x}_{12}(z_1-z_2) R^{y}_{23}(z_2-z_3) - R^{y}_{13}(z_1-z_3) R^{x-y}_{12}(z_1-z_2) - R^{y-x}_{23}(z_2-z_3) R^{x}_{13}(z_1-z_3)=
    }
    \\ \ \\
    \displaystyle{
    =\frac{\pi}{2\cos(\frac{\pi x}{2})\cos(\frac{\pi y}{2})\cos(\frac{\pi (x-y)}{2})}\sum_{a\neq b\neq c\neq a}^{N+M} e_{aa}\otimes e_{bb}\otimes e_{cc}\,,
    }
\end{array}\eq
where the right-hand side is independent of the spectral parameters.
For $N+M\leq 2$ (\ref{relAY}) turns into (\ref{AYBE}).
\end{theor}
The proofs of both theorems are given in Appendix A. The result of
Theorem \ref{th1} extends the classification \cite{SchP} of trigonometric solutions of (\ref{AYBE}) to
the graded case. As we previously mentioned, the non-supersymmetric version of (\ref{superR})
has generalizations satisfying (\ref{AYBE}). For example, it is true for the 7-vertex trigonometric $R$-matrix and its
higher rank analogues.
The question whether these type $R$-matrices have supersymmetric
counterparts deserves further elucidation.
Different applications of the associative Yang-Baxter equation (see \cite{FK,SchP,LOZ14,KrZ,GZ,SeZ,SeZ2}) can be
studied for the graded $R$-matrix (\ref{superR}).


  \subsection{Quantum Yang-Baxter equation}

We first mention that both $R$-matrices (\ref{uqsuper}) and (\ref{superR})
obey the {\em unitarity property}
\beq\label{q03}\begin{array}{c}\displaystyle
    R^{\hbar}_{12}(z) R^\hbar_{21}(-z)= \phi(\hbar,z)\phi(\hbar,-z)\rm{Id}=\Big( \frac{\pi^2}{\sin^2(\pi \hbar)}-
  \frac{\pi^2}{\sin^2(\pi z)}\Big)\rm{Id}\,.
\end{array}\eq
and the {\em skew-symmetry property}
\beq\label{q031}\begin{array}{c}\displaystyle
    R^{-\hbar}_{12}(-z)=-R^\hbar_{21}(z)=-P_{12}R^\hbar_{12}(z)P_{12}\,.
\end{array}\eq
This is simply verified by straightforward calculation.

\begin{theor}\label{th3}
The supersymmetric extension of the $\mZ_n$-invariant $R$-matrix (\ref{superR}) satisfies the quantum
Yang-Baxter equation (\ref{QYB}).
\end{theor}
\begin{proof}
Of course, one can verify the statement by direct computation. But there is a  simpler way to do it
using the statement of Theorem \ref{th1}. It is known (see e.g. Section 4 in \cite{LOZ14}) that
any unitary (\ref{q03}) and skew-symmetric (\ref{q031}) solution of the associative Yang-Baxter equation
 (\ref{AYBE}) is also a solution of the quantum Yang-Baxter equation (\ref{QYB}). The $R$-matrix (\ref{superR})
satisfies (\ref{AYBE}) and obeys the properties (\ref{q03}) and (\ref{q031}). This finishes the proof.
\end{proof}

To summarize, both $R$-matrices (\ref{uqsuper}) and (\ref{superR})  satisfy
the quantum Yang-Baxter equation (\ref{QYB}) and obey the properties (\ref{q03}), (\ref{q031}).

We use two different normalizations of $R$-matrices. The one (\ref{q03}) means that $R$-matrix is normalized
to the function $\phi(\hbar,z)$.
$R$-matrices normalized as in (\ref{q03})
are used in the
the $R$-matrix identities. Another normalization
is the standard one.
The normalized $R$-matrices are defined as follows:
\beq\label{q04}
\begin{array}{c}
\displaystyle{
    {\bar R}^{\hbar}_{12}(z)  = \frac{1}{\phi(\hbar,z)}\,R^{\hbar}_{12}(z)=\frac{1}{\pi\cot(\pi \hbar)+\pi\cot(\pi z)}\, R^{\hbar}_{12}(z)\,.
    }
\end{array}\eq
Then (\ref{q03}) takes the form:
\beq\label{q05}
\begin{array}{c}
    {\bar R}^{\hbar}_{12}(z) {\bar R}^\hbar_{21}(-z)= {\rm Id}\,.
\end{array}
\eq
A list of the normalized $R$-matrices (\ref{uqsuper})-(\ref{g1,1}) is given in Appendix B.

The short notation ${\bar R}_{ij}={\bar R}^{\hbar}_{ij}(z_i-z_j)$ is used for the $R$-matrix acting non-trivially on the $i$-th and $j$-th tensor components of the Hilbert
space $\mH=(\mC^{N|M})^{\otimes L}$. The quantum Yang-Baxter equation (\ref{QYB}) implies
\beq\label{QYB2}
\begin{array}{c}
\displaystyle{
    {\bar R}^{\hbar}_{ij}(z_i-z_j)  {\bar R}^{\hbar}_{ik}(z_i-z_k) {\bar R}^{\hbar}_{jk}(z_j-z_k) =
      {\bar R}^{\hbar}_{jk}(z_j-z_k) {\bar R}^{\hbar}_{ik}(z_i-z_k) {\bar R}^{\hbar}_{ij}(z_i-z_j)
      }
\end{array}\eq
for any distinct integers $1\leq i,j,k\leq N$ and
\beq\label{QYB3}
\begin{array}{c}
\displaystyle{
    [{\bar R}^{\hbar}_{ij}(u), {\bar R}^{\hbar'}_{kl}(v)]=0
      }
\end{array}\eq
for any distinct integers $1\leq i,j,k,l\leq N$.

\subsection{Twist transformation}
Here we describe a relation between the ${\hat{\mathcal U}}_q({\rm gl}(N|M))$ $R$-matrix $\mR^\hbar_{12}(u-v)$ in (\ref{uqsuper}), and the one in (\ref{superR}),  $R^\hbar_{12}(u-v)$.
It is given by the following Drinfeld twist.

\begin{theor}\label{th4}
The relation between the $R$-matrices (\ref{uqsuper})  and (\ref{superR}) is as follows:
\beq\label{b01}
\displaystyle{
R^\hbar_{12}(u-v)=G_1(u)G_2(v)F_{12}(\hbar)\mR^\hbar_{12}(u-v)F_{21}^{-1}(\hbar)G_1^{-1}(u)G_2^{-1}(v),
}
\eq
where $G_1(u)=G(u)\otimes 1_{N+M}$, $G_2(v)=1_{N+M}\otimes G(v)$,
\beq\label{b02}
\displaystyle{
F_{12}(\hbar)=\sum_{i,j=1}^{N+M} \exp\left(\frac{\pi \imath \hbar \left(2(i-j)-(N+M){\rm sign}(i-j)\right)}{2(N+M)}\right)e_{ii}\otimes e_{jj}\,,
}
\eq
\beq\label{b021}
\displaystyle{
{\rm sign}(i-j)=
\left\{
\begin{array}{cc}
1,& i>j\,,
\\
0,& i=j\,,
\\
-1,& i<j
\end{array}
\right.
}
\eq
and
\beq\label{b03}
\displaystyle{
G(u)=\sum_{j=1}^{N+M} \exp\left( \frac{2\pi \imath  (j-1)u}{N+M} \right)e_{jj}\,.
}
\eq
\end{theor}
This statement may be confirmed by direct computation.

Notice that for the case $N+M=2$ the twist matrix $F$ becomes trivial: $F_{12}(\hbar)=1_{N+M}^{\otimes 2}$. That is, the
$R$-matrices (\ref{2,0}) and (\ref{g2,0}) as well as (\ref{1,1}) and (\ref{g1,1}) are related
by a gauge transformation with $G(u)={\rm diag}(1,e^{\pi\imath u})$. Obviously, the relation (\ref{b01})
remains the same for the normalized $R$-matrices from Appendix B.

\section{Graded spin Ruijsenaars-Macdonald operators}\label{sec3}
\setcounter{equation}{0}

Let $p_i$ be the shift operator defined by its action on a (sufficiently smooth) function $f(z_1,\dots ,z_L)$:
\beq\label{p_i}
(p_if)(z_1,z_2,\dots z_L)=\exp\left(-\eta \frac{\partial}{\partial z_i}\right)f(z_1,\dots,z_L)=f(z_1,\dots,z_i-\eta,\dots, z_L)\,.
\eq
Denote by $D_k$ the trigonometric Ruijsenaars-Macdonald operators \cite{Ruij,Macd}:
\begin{equation}\label{Dscalar}
 D_k=\sum\limits_{\substack{|I|=k}}\prod\limits_{\substack{i\in I \\ j\notin I}}\phi(z_j-z_i)\prod_{i\in I}p_{i},\qquad k=1,\dots,L,
\end{equation}
 where the sum is taken over all subsets $I$ of $\{1,\dots,L\}$ of size $k$, and
 \beq\label{a081}
 \begin{array}{c}
  \displaystyle{
\phi(z)\equiv\phi(\hbar,z)\stackrel{(\ref{a0811})}{=}\pi\cot(\pi z)+\pi\cot(\pi \hbar)=\frac{\pi\sin(\pi(z+\hbar))}{\sin(\pi z)\sin(\pi\hbar)}\,.
 }
 \end{array}
 \eq
 The commutativity of operators (\ref{Dscalar}) holds:
 \beq\label{comm}
  \begin{array}{c}
  \displaystyle{
   [D_k,D_l]=0\quad k,l=1,...,L\,.
 }
 \end{array}
 \eq
 Namely, it was shown in  \cite{Ruij} that the commutativity is equivalent to a set of
 relations for the function $\phi$, which can be considered as a set of functional equations.
 For the function (\ref{a0811}) these relations become identities.

Following ideas of \cite{Uglov,Lam2} we introduced in \cite{MZ1} a set of (spin or matrix-valued) difference
 operators. The graded
(spin) Ruijsenaars-Macdonald operators are defined in the same way (but with the graded $R$-matrices):
\beq\label{q10}
\begin{array}{c}
  \displaystyle{
    {\mathcal D}_k=\sum\limits_{1\leq i_1<...<i_k\leq L}\left(\!\prod\limits^{L}_{\hbox{\tiny{$ \begin{array}{c}{ j=1 }\\{ j\!\neq\! i_1...i_{k-1} } \end{array}$}}}\!\phi(z_j-z_{i_1})\ \phi(z_j-z_{i_2})
    \ \cdots\
    \phi(z_j-z_{i_k})\right)\times
    }
    \\ \ \\
    \displaystyle{
    \times\left(
   \overleftarrow{\prod\limits_{j_1=1}^{i_1-1}} \bar{R}_{j_1 i_1}
   \overleftarrow{\prod\limits^{i_2-1}_{\hbox{\tiny{$ \begin{array}{c}{ j_2=1 }\\{ j_2\!\neq\! i_1 } \end{array}$}}}} \bar{R}_{j_2 i_2}
      \ \ldots\
 \overleftarrow{\prod\limits^{i_k-1}_{\hbox{\tiny{$ \begin{array}{c}{ j_k=1 }\\{ j_k\!\neq\! i_1...i_{k-1} } \end{array}$}}}} \bar{R}_{j_k i_k}
  \right)\times
   }
   \\ \ \\
     \displaystyle{
       \times p_{i_1}\cdot p_{i_2}\cdots p_{i_k}\times\left(
   \overrightarrow{\prod\limits^{i_k-1}_{\hbox{\tiny{$ \begin{array}{c}{ j_k\!=\!1 }\\{ j_k\!\neq\! i_{1}...i_{k-1}} \end{array}$}}}}\bar{R}_{i_k j_k}
   \overrightarrow{\prod\limits^{i_{k-1}-1}_{\hbox{\tiny{$ \begin{array}{c}{ j_{k-1}\!=\!1 }\\{ j_{k-1}\!\neq\! i_{1}...i_{k-2}} \end{array}$}}}}\bar{R}_{i_{k-1} j_{k-1}}
      \ \ldots\
  \overrightarrow{\prod\limits^{i_{1}-1}_ {j_1=1}} \bar{R}_{i_{1} j_{1}}\right),
 }
\end{array}\eq
where $k=1,...,L$, $\bar{R}_{ij}=\bar{R}_{ij}^\hbar(z_i-z_j)$.
 The arrows in (\ref{q10}) mean the ordering in the $R$-matrix products. For instance,
$\overrightarrow{\prod\limits^{L}_ {j=1}} R_{ij}=R_{i,1}R_{i,2}...R_{i,L}$ and $\overleftarrow{\prod\limits^{L}_ {j=1}} R_{ji}=R_{L,i}R_{L-1,i}...R_{1,i}$. See \cite{MZ1} for details.

It was shown in \cite{MZ1} that the operators ${\mathcal D}_k$ in (\ref{q10}) commute with each other
 \beq\label{a205}
  \begin{array}{c}
  \displaystyle{
   [{\mathcal D}_m,{\mathcal D}_l]=0\quad m,l=1,...,L
 }
 \end{array}
 \eq
 if and only if the following set of identities holds true:
 \beq\label{a20}
  \begin{array}{c}
  \displaystyle{
   \sum\limits_{1\leq i_1<...<i_k\leq L}
    \Big({\mathcal F}^-_{i_1,...,i_k}(k,L)-{\mathcal F}^+_{i_1,...,i_k}(k,L)\Big)=0\,,
 }
 \end{array}
 \eq
where
 \beq\label{a21}
  \begin{array}{c}
  \displaystyle{
   {\mathcal F}^+_{i_1,...,i_k}(k,L)=
   \overrightarrow{\prod\limits_{l_k=i_k+1}^L} R_{i_k l_k}
   \overrightarrow{\prod\limits^L_{\hbox{\tiny{$ \begin{array}{c}{ l_{k-1}\!=\!i_{k-1}\!+\!1 }\\{ l_{k-1}\!\neq\! i_k } \end{array}$}}}}R_{i_{k-1} l_{k-1}}
      \ \ldots\
   \overrightarrow{\prod\limits^L_{\hbox{\tiny{$ \begin{array}{c}{ l_{1}\!=\!i_{1}\!+\!1 }\\{ l_{1}\!\neq\! i_2...i_k } \end{array}$}}}}R_{i_{1} l_{1}}
   \times
   }
   \end{array}
 \eq
   $$
   \begin{array}{c}
     \displaystyle{
 \times
 \overleftarrow{\prod\limits^L_{\hbox{\tiny{$ \begin{array}{c}{ j_{1}\!=\!1 }\\{ j_{1}\!\neq\! i_1...i_k } \end{array}$}}}}R_{j_1i_1}^-
 \overleftarrow{\prod\limits^L_{\hbox{\tiny{$ \begin{array}{c}{ j_{2}\!=\!1 }\\{ j_{2}\!\neq\! i_1...i_k } \end{array}$}}}}R_{j_2i_2}^-
 \ \ldots\
  \overleftarrow{\prod\limits^L_{\hbox{\tiny{$ \begin{array}{c}{ j_{k}\!=\!1 }\\{ j_{k}\!\neq\! i_1...i_k } \end{array}$}}}}R_{j_ki_k}^-
  \times
    }
   \\ \ \\
     \displaystyle{
       \times
   \overrightarrow{\prod\limits^{i_k-1}_{\hbox{\tiny{$ \begin{array}{c}{ m_k\!=\!1 }\\{ m_k\!\neq\! i_{1}...i_{k-1}} \end{array}$}}}}R_{i_k m_k}
   \overrightarrow{\prod\limits^{i_{k-1}-1}_{\hbox{\tiny{$ \begin{array}{c}{ m_{k-1}\!=\!1 }\\{ m_{k-1}\!\neq\! i_{1}...i_{k-2}} \end{array}$}}}}R_{i_{k-1} m_{k-1}}
      \ \ldots\
  \overrightarrow{\prod\limits^{i_{1}-1}_ {m_1=1}} R_{i_{1} m_{1}}
 }
 \end{array}
 $$
and
 \beq\label{a22}
  \begin{array}{c}
  \displaystyle{
   {\mathcal F}^-_{i_1,...,i_k}(k,L)=
   \overleftarrow{\prod\limits_{m_1=1}^{i_1-1}} R_{m_1 i_1}
   \overleftarrow{\prod\limits^{i_2-1}_{\hbox{\tiny{$ \begin{array}{c}{ m_2=1 }\\{ m_2\!\neq\! i_1 } \end{array}$}}}} R_{m_2 i_2}
      \ \ldots\
 \overleftarrow{\prod\limits^{i_k-1}_{\hbox{\tiny{$ \begin{array}{c}{ m_k=1 }\\{ m_k\!\neq\! i_1...i_{k-1} } \end{array}$}}}} R_{m_k i_k}
   \times
   }
   \end{array}
 \eq
   $$
   \begin{array}{c}
     \displaystyle{
 \times
 \overrightarrow{\prod\limits^L_{\hbox{\tiny{$ \begin{array}{c}{ j_{k}\!=\!1 }\\{ j_{k}\!\neq\! i_1...i_k } \end{array}$}}}}R_{i_k j_k}^-
 \overrightarrow{\prod\limits^L_{\hbox{\tiny{$ \begin{array}{c}{ j_{k-1}\!=\!1 }\\{ j_{k-1}\!\neq\! i_1...i_k } \end{array}$}}}}R_{i_{k-1} j_{k-1}}^-
 \ \ldots\
  \overrightarrow{\prod\limits^L_{\hbox{\tiny{$ \begin{array}{c}{ j_{1}\!=\!1 }\\{ j_{1}\!\neq\! i_1...i_k } \end{array}$}}}}R_{i_1 j_1}^-
  \times
    }
   \\ \ \\
     \displaystyle{
       \times
   \overleftarrow{\prod\limits^{L}_{\hbox{\tiny{$ \begin{array}{c}{ l_1\!=\!i_1\!+\!1 }\\{ l_1\!\neq\! i_{2}...i_{k}} \end{array}$}}}}R_{l_1 i_1}
    \overleftarrow{\prod\limits^{L}_{\hbox{\tiny{$ \begin{array}{c}{ l_2\!=\!i_2\!+\!1 }\\{ l_2\!\neq\! i_{3}...i_{k}} \end{array}$}}}}R_{l_2 i_2}
      \ \ldots\
  \overleftarrow{\prod\limits^{L}_ {l_k=i_k+1}} R_{l_k i_k}\,,
 }
 \end{array}
 $$
and $R_{ij}=R_{ij}^\hbar(z_i-z_j)$, $R_{ij}^-=R_{ij}^\hbar(z_i-z_j-\eta)$. In the scalar case $M=1$, the identities (\ref{a20})
coincide with the identities for $\phi$-functions underlying commutativity of the Ruijsenaars-Macdonald operators
(\ref{Dscalar}).

The expressions  ${\mathcal D}_k$ (\ref{q10}) are  matrix-valued difference operators.
For example,
\beq\label{D1N}
\mathcal{D}_1=\sum_{i=1}^L \Big(\prod\limits_{\substack{j=1\\j\neq i}}^L\phi(z_j-z_i)\Big)\bar{R}_{i-1,i}\bar{R}_{i-2,i}\dots \bar{R}_{1,i}p_i\, \bar{R}_{i,1}\dots \bar{R}_{i,i-2}\bar{R}_{i,i-1}\,.
\eq
That is, ${\mathcal D}_k$ in (\ref{q10}) are ${\rm End}(\mH)$-valued
difference operators.
The commutativity of operators (\ref{q10}) was proved (in the case $M=0$) for the trigonometric ${\hat{\mathcal U}}_q({\rm gl}_2)$ $R$-matrix
 in \cite{Lam2}. In \cite{MZ1} a proof was given based on the identities (\ref{a20}) for the elliptic $R$-matrix (including some trigonometric and rational degenerations). The ${\hat{\mathcal U}}_q({\rm gl}_n)$ $R$-matrix
 was considered separately in \cite{MZ2}.

In order to prove that commutativity of operators (\ref{q10}) holds true for the
graded $R$-matices (\ref{uqsuper}) and (\ref{superR}) we need to prove the identities (\ref{a20}) for these
$R$-matrices. In fact, the proof is  almost the same as in \cite{MZ1,MZ2}.
 Let us give a sketch of the proof. The aim is to show that the l.h.s. of (\ref{a20}) (denote it by $\mathcal F$) is independent of $\eta$. Then we can put $\eta=0$ and the statement becomes simple: $\mathcal F|_{\eta=0}=0$ due to the unitarity property (\ref{q03}) of the $R$-matrix. More precisely, $\mathcal F|_{\eta=0}={\rm Id}\mathcal F_{\rm scalar}$,
 where $\mathcal F_{\rm scalar}$ is the same expression as $\mathcal F$ with $R$-matrices being replaced by $\phi$-functions (\ref{a081}). In this case the statement ($\mathcal F_{\rm scalar}=0$) was proved in \cite{Ruij}.
   In order to prove that  the l.h.s. of (\ref{a20}) is a constant as a function of $\eta$, we make two steps. The first step is to show that $\mF$ has no poles in the variable $\eta=z_i-z_j$. The proof is given  in \cite{MZ1} (see Proposition 4.1) and can be performed in the same way for the graded $R$-matrices with the property $\res_{u=0} R_{12}(u)=P_{12}$, where $P_{12}$ is  the graded permutation. The second step is to show that the function $\mF$ has no poles at $\eta=z_i-z_j+m$, where $m\in\mZ$. For this purpose we use the quasi-periodic property of $R$-matrix (\ref{superR}):
 \beq\label{r72}
 \begin{array}{c}
  \displaystyle{
 R_{12}^\hbar(z+1)=(Q\otimes 1_{N+M})R_{12}^\hbar(z)(Q^{-1}\otimes 1_{N+M})\,,
  }
 \end{array}
 \eq
where
\beq\label{r73}
Q=\sum_{j=1}^{N+M} \exp\left(\frac{2\pi \imath j}{N+M}\right)e_{jj}\,.
\eq
Notice that $Q$ is independent of the spectral parameter $z$. Therefore, we have
\beq
\res\limits_{z_i=z_j+\eta+m}\mF=(Q^m\otimes 1_{N+M})\res\limits_{z_i=z_j+\eta}\mF\ \,(Q^{-m}\otimes 1_{N+M})=0\,.
\eq
The absence of poles at $\eta=z_i-z_j+m$  guarantees that $\mF$ is a constant function of $\eta$, see Appendix C in \cite{MZ2}. This proof works also for the case of the ${\hat{\mathcal U}}_q({\rm gl}(N|M))$-valued $R$-matrix (\ref{uqsuper}), which is periodic: $\mathcal R_{12}^{\hbar} (z+1)=\mathcal R_{12}^{\hbar} (z)$. In fact, the proof in this case repeats
the one from Appendix C in \cite{MZ2} for the non-supersymmetric case. The only difference is that the permutation operator is now graded.

We also mention that the identity (\ref{a20}) for $k=1$ follows
from the associative Yang-Baxter equation (see \cite{MZ1}), so that for the $R$-matrix (\ref{superR}) the relation
(\ref{a20}) with $k=1$ follows from (\ref{AYBE}). Another possible application of the associative Yang-Baxter equation
is in the construction of classical analogues of integrable models related to the presented difference operators.
These are models of interacting tops. The underlying Lax representations and $R$-matrix structures are discussed in  \cite{KrZ,GZ,SeZ,SeZ2}.

\section{Long-range spin chains}\label{sec5}
\setcounter{equation}{0}

Here we derive commuting Hamiltonians of long-range spin chains by applying
the Polychronakos freezing trick \cite{Polych} to the difference operators (\ref{q10}).
The derivation is performed similarly to the one presented in \cite{MZ2}.

\subsection{Hamiltonians of the spin chain}
Consider expansion of the scalar difference operators $D_k$ (\ref{Dscalar}) in the variable $\eta$ near $\eta=0$:
\begin{equation}\label{s29}
\displaystyle{
  D_k=D_k^{[0]}+\eta D_k^{[1]}+\eta^2 D_k^{[2]}+O(\eta^3)\,,\quad k=1,...,L\,.
  }
\end{equation}
Using also the expansion of (\ref{p_i})
\beq\label{s291}
p_i=\sum_{k=0}^m\frac{(-1)^m}{m!}\,\eta^m\frac{\partial^m}{\partial z_i^m}+{O}(\eta^{m+1})\,,
\eq
we  conclude that $D_k^{[m]}$ are some $m$-th order differential operators. For example,
\begin{equation}\label{s30}
\displaystyle{
  D_k^{[0]}=\sum_{|I|=k}\prod\limits_{\substack{i\in I\\j\notin I}}\phi(z_j-z_i)\,,\quad k=1,...,L
  }
\end{equation}
and
\begin{equation}\label{s31}
\displaystyle{
  D_k^{[1]}=-\sum\limits_{m=1}^L\Big(\sum\limits_{\substack{|I|=k\\ m\in I}}\prod\limits_{\substack{i\in I\\j\notin I}}\phi(z_j-z_i)\Big)\p_{z_m}\,.
  }
\end{equation}
Consider now a similar expansion of the spin operators (\ref{q10})  $\mD_k$ in variable $\eta$ (near $\eta=0$):
\begin{equation}\label{s34}
\displaystyle{
  \mD_k=\mD_k^{[0]}+\eta \mD_k^{[1]}+\eta^2 \mD_k^{[2]}+O(\eta^3)\,,\quad k=1,...,L\,.
  }
\end{equation}
Since $\mD_k^{[0]}=\mD_k|_{\eta=0}$, due to the unitarity (\ref{q05}), we have
\beq\label{s35}
\mathcal{D}_k^{[0]}={\rm Id}\sum_{|I|=k}\prod\limits_{\substack{i\in I\\j\notin I}}\phi(z_j-z_i)={\rm Id}\,D_k^{[0]}\,,
\eq
where ${\rm Id}=1_{N+M}^{\otimes L}$ is the identity matrix in ${\rm End}(\mH)$.
For the set of $\mD_k^{[1]}$ one gets
\beq\label{s36}
\begin{array}{c}
\displaystyle{
-\mathcal{D}_1^{[1]}={\rm Id}\sum_{i=1}^L \prod\limits_{\substack{j=1\\j\neq i}}^L\phi(z_j-z_i)\frac{\partial}{\partial z_i}+
}
\\
\displaystyle{
+\sum_{i=1}^L \prod\limits_{\substack{j=1\\j\neq i}}^L\phi(z_j-z_i)\sum_{k=1}^{i-1}\bar{R}_{i-1,i}\dots \bar{R}_{k+1,i}\bar{R}_{k,i}\left(\frac{\partial}{\partial z_i}\bar{R}_{i,k}\right)\bar{R}_{i,k+1}\dots \bar{R}_{i,i-1}\,.
}
\end{array}
\eq
In the general case we also have
\beq\label{s39}
\displaystyle{
\mathcal{D}_k^{[1]}={\rm Id}\, D_k^{[1]}-{\tilde H}_k\,,\quad k=1,...,L\,,
}
\eq
where ${\tilde H}_k\in{\rm End}(\mH)$ are some matrix-valued functions, which contain $R$-matrix derivatives but do not contain differential operators.

By restricting the operators ${\tilde H}_k$ to the points $z_k=x_k=k/L$
we obtain the set of Hamiltonians of long-range spin chains:
\begin{equation}\label{s431}
\displaystyle{
 H_i={\tilde H}_i\Big|_{z_k=x_k=k/L}\,,\quad i=1,...,L-1\,.
  }
\end{equation}
Commutativity of these Hamiltonians
\begin{equation}\label{s45}
\displaystyle{
 [H_i,H_j]=0\,,\quad i,j=1,...,L-1
  }
\end{equation}
follows from the set of identities for the function $\phi$:
 \beq \label{t604}
\sum\limits_{\substack{|I|=k\\ l\in I}}\prod\limits_{\substack{i\in I\\j\notin I}}\phi(x_j-x_i)=\sum\limits_{\substack{|I'|=k\\ m\in I'}}\prod\limits_{\substack{i\in I'\\j\notin I'}}\phi(x_j-x_i) \qquad \text{for  } l,m=1\dots N\,,
\eq
and for any $k=1,...,L$. See \cite{MZ2} for the proof in the elliptic case. This proof is the same for the
trigonometric case.
The set of equidistant points $x_k$ is an equilibrium position in the underlying classical spinless model which is the trigonometric (spinless) Ruijsenaars-Schneider model. Moreover, it is the equilibrium position for all
flows of this model.

Introduce the following compact notations:
\beq\label{s4601}
\begin{array}{c}
\displaystyle{
{\bar F}^\hbar_{ij}(z)=\frac{\partial}{\partial z}\bar{R}^\hbar_{ij}(z)\,,\qquad
{\bar R}_{ij}={\bar R}^\hbar_{ij}(x_i-x_j)
}
\end{array}
\eq
and
\beq\label{s4602}
\begin{array}{c}
\displaystyle{
{\bar F}_{ij}={\bar F}^\hbar_{ij}(x_i-x_j)\,.
}
\end{array}
\eq
Let us write down expressions for the first two Hamiltonians.
\paragraph{Example. The first Hamiltonian:}
\beq\label{s461}
\begin{array}{c}
\displaystyle{
H_1 =\sum\limits_{k<i}^L \bar{R}_{i-1,i}\dots \bar{R}_{k+1,i}\bar{R}_{k,i}\bar{F}_{i,k}\bar{R}_{i,k+1}\dots \bar{R}_{i,i-1}\,.
}
\end{array}
\eq
\paragraph{Example. The second Hamiltonian:}
\beq\notag
\begin{array}{c}
\displaystyle{
H_2=\sum\limits_{\substack{m,l=1\\m<l}}^L \prod\limits_{\substack{j=1\\j\neq m,l}}^L\phi(x_j-x_m)\phi(x_j-x_l)\times
}
\\
\displaystyle{
\times\left(\sum_{i=1}^{m-1}\bar{R}_{m-1,m}\dots \bar{R}_{i+1,m}\bar{R}_{i,m}\bar{F}_{m,i}\bar{R}_{m,i+1}\dots \bar{R}_{m,m-1}+\right.
}
\end{array}
\eq
\beq\label{s47}
\begin{array}{c}
\displaystyle{
+\sum_{i=1}^{m-1}\bar{R}_{m-1,m}\dots \bar{R}_{1,m}\bar{R}_{l-1,l}\dots \bar{R}_{m+1,l}\bar{R}_{m-1,l}\dots\bar{R}_{i+1,l}\bar{R}_{i,l}\times}
\\ \ \\
{\displaystyle\times\bar{F}_{l,i}\bar{R}_{l,i+1}\dots\bar{R}_{l,m-1}\bar{R}_{l,m+1}\dots \bar{R}_{l,l-1}\bar{R}_{m,1}\dots \bar{R}_{m,m-1}+
}
\end{array}
\eq
\beq\notag
\begin{array}{c}
\displaystyle{
\left.+\sum_{i=m+1}^{l-1}\bar{R}_{l-1,l}\dots\bar{R}_{i+1,l}\bar{R}_{i,l}\bar{F}_{l,i}\bar{R}_{l,i+1}\dots \bar{R}_{l,l-1}\right)\,.
}
\end{array}
\eq

\subsection{Supersymmetric q-deformed Haldane-Shastry Hamiltonian}
For the normalized $R$-matrix (\ref{1,1n}), direct calculation provides
the following relation:
\beq
\label{1,11}
   \begin{array}{c}
   \displaystyle{
  \bar{R}^\hbar_{12}(v-u)\bar{F}_{21}^\hbar(u-v)=
   - \pi\frac{\sin(\pi \hbar)}{\sin \pi(\hbar+u-v)\sin\pi(\hbar-u+v)}\, C^{\rm susy}_{12}\,,
  }
  \end{array}
\eq
where $C^{\rm susy}_{12}$ does not depend on $u,v$:
\beq
  C^{\rm susy}_{12}=\displaystyle{
 \left(e^{-\imath\pi \hbar} e_{11}\otimes e_{22}+e_{12}\otimes e_{21}-e_{21}\otimes e_{12}+e^{\imath\pi \hbar}e_{22}\otimes e_{11}+2\cos(\pi \hbar) e_{22}\otimes e_{22}\right)}\,.
  \eq
Then the Hamiltonian (\ref{s461}) takes the form:
  \beq\label{s461re}
\begin{array}{c}
\displaystyle{
{ H}_1=-\pi\sum\limits_{k<i}^L \frac{\sin(\pi \hbar)}{\sin \pi(\hbar+x_i-x_k)\sin\pi(\hbar-x_i+x_k)}\, \bar{R}_{i-1,i}\dots \bar{R}_{k+1,i}C^{\rm susy}_{k,i}\bar{R}_{i,k+1}\dots \bar{R}_{i,i-1}\,;
}
\end{array}
\eq
that is, it has the same form as the non-supersymmetric one (\ref{s461re0}) but
the matrix $C_{12}$ (\ref{s461re01}) is replaced with $C^{\rm susy}_{12}$ (\ref{Csusy}).

\subsection{Non-relativistic limit to the graded Haldane-Shastry model}
The non-relativistic limit
is the limit $\hbar\to 0$ (see remark between (\ref{2,0}) and (\ref{1,1})). For example, it reproduces the Haldane-Shastry model (\ref{HS})
from the q-deformed case (\ref{s461re0})-(\ref{s461re01}) just as the Calogero-Sutherland model is obtained
from the Ruijsenaars-Schneider one.
 For the graded case in this limit we have
  $\bar{R}^\hbar_{ij}\to {\rm Id}$ and  $C^{\rm susy}_{ij}\to (1-P_{ij})$. Thus,
\beq\label{q400}
\lim_{\hbar\to 0} \frac{1}{\hbar}{ H}_1=\pi^2\sum\limits_{k<i}^L
\frac{1}{\sin^2(\pi(x_i-x_k))}\,(1-P_{ik})\,,
\eq
which coincides with the Hamiltonian of (the graded) Haldane-Shastry spin chain (\ref{HS}).

Consider the $R$-matrix (\ref{g1,1n}). Upon substitution into
(\ref{s461}) it provides the q-deformed supersymmetric
version of the anisotropic model (\ref{s0911}). By taking the non-relativistic limit
 we obtain the anisotropic spin chain with the following Hamiltonian:
\beq\label{q41}
\begin{array}{c}
\displaystyle{
\lim_{\hbar\to 0} \frac{1}{\hbar}{ H}_1=\pi^2\sum\limits_{k<i}^L
\frac{1}{\sin^2(\pi(x_i-x_k))}\,\Big(e_{11}\otimes e_{22}+e_{22}\otimes e_{11}+2 e_{22}\otimes e_{22}\Big)+
}
\\ \ \\
\displaystyle{
+\pi^2\sum_{k<i}^L
\frac{\cos (\pi(x_i-x_k))}{\sin^2(\pi(x_i-x_k))}\,\Big(e_{12}\otimes e_{21}-e_{21}\otimes e_{12}\Big)\,.
}
\end{array}
\eq


\section{Appendix A: proof of the associative Yang-Baxter equation}
\setcounter{equation}{0}
\def\theequation{A.\arabic{equation}}
%
\paragraph{Proof of Theorem \ref{th1}.}
Introduce the following notations
\beq
f^a(z,\hbar)=\pi\Big( (-1)^{p_a}\cot(\pi z)+\cot(\pi\hbar)\Big)\,,
\eq
\beq
g^{ac}(z)=\frac{\pi}{\sin(\pi z)} \exp\Big(\frac{\pi\imath z}{N+M}\Big(2(a-c)-(N+M){\rm sign}(a-c)\Big)\Big)\,.
\eq
 Then the $R$-matrix (\ref{superR}) takes the form:
\beq
  R^\hbar_{12}(z)
  =\sum\limits_a f^a(z,\hbar) e_{aa}\otimes e_{aa}
  +\sum\limits_{a\neq c} g^{ac}(\hbar) e_{aa}\otimes e_{cc}
 +\sum\limits_{a\neq c} (-1)^{p_c} g^{ac}(z) e_{ac}\otimes e_{ca}.
\eq
For the left-hand side of (\ref{AYBE}) we have:
\beq\label{A1}
\begin{array}{c}
\displaystyle{
 R^{x}_{12}(z_{12}) R^{y}_{23}(z_{23})=\sum\limits_a f^a(z_{12},x)f^{a}(z_{23},y) e_{aa}\otimes e_{aa}\otimes e_{aa}+
  }
      \\ \ \\
      \displaystyle{
 +\sum\limits_{a\neq c} f^a(z_{12},x) g^{ac}(y) e_{aa}\otimes e_{aa}\otimes e_{cc}
     + \sum\limits_{a\neq c}(-1)^{p_c} f^a(z_{12},x) g^{ac}(z_{23}) e_{aa}\otimes e_{ac}\otimes e_{ca}
      }
      \\ \ \\
      \displaystyle{
      + \sum\limits_{a\neq c} g^{ac}(x)f^c(z_{23},y) e_{aa}\otimes e_{cc}\otimes e_{cc}+\sum\limits_{\substack{a\neq c \\ b\neq c}} g^{ac}(x)g^{cb}(y) e_{aa}\otimes e_{cc}\otimes e_{bb}
       }
\end{array}
\eq
$$
\begin{array}{c}
      \displaystyle{
      + \sum\limits_{\substack{a\neq c \\ b\neq c}} (-1)^{p_b}g^{ac}(x)g^{cb}(z_{23}) e_{aa}\otimes e_{cb}\otimes e_{bc}+
     + \sum\limits_{a\neq c}(-1)^{p_c} g^{ac}(z_{12})f^a(z_{23},y) e_{ac}\otimes e_{ca}\otimes e_{aa}
      }
      \\ \ \\
      \displaystyle{
      +\sum\limits_{\substack{a\neq c \\ b\neq a}}(-1)^{p_c} g^{ac}(z_{12})g^{ab}(y) e_{ac}\otimes e_{ca}\otimes e_{bb}
      +\sum\limits_{\substack{a\neq c \\ b\neq a}}(-1)^{p_c+p_b} g^{ac}(z_{12})g^{ab}(z_{23}) e_{ac}\otimes e_{cb}\otimes e_{ba}.
     }
\end{array}
$$
Next, write down the first item on right-hand side of (\ref{AYBE}):
\beq\label{A2}
\begin{array}{c}
\displaystyle{
  R^{y}_{13}(z_{13}) R^{x-y}_{12}(z_{12})=\sum\limits_a f^a(z_{13},y)f^{a}(z_{12},x-y) e_{aa}\otimes e_{aa}\otimes e_{aa}+
  }
  \\ \ \\
      \displaystyle{
  +\sum\limits_{a\neq c} f^a(z_{13},y) g^{ac}(x-y) e_{aa}\otimes e_{cc}\otimes e_{aa}
     + \sum\limits_{a\neq c}(-1)^{p_c} f^a(z_{13},y) g^{ac}(z_{12}) e_{ac}\otimes e_{ca}\otimes e_{aa}
     }
       \\ \ \\
      \displaystyle{
     +\sum\limits_{a\neq c}g^{ac}(y) f^a(z_{12},x-y) e_{aa}\otimes e_{aa}\otimes e_{cc}
     +\sum\limits_{\substack{a\neq c \\ a\neq b}}g^{ac}(y) g^{ab}(x-y) e_{aa}\otimes e_{bb}\otimes e_{cc}
     }
\end{array}
\eq
$$
\begin{array}{c}
      \displaystyle{
     +\sum\limits_{\substack{a\neq c \\ a\neq b}}(-1)^{p_b} g^{ac}(y) g^{ab}(z_{12}) e_{ab}\otimes e_{ba}\otimes e_{cc}
     +\sum\limits_{a\neq c}(-1)^{p_c}g^{ac}(z_{13}) f^c(z_{12},x-y) e_{ac}\otimes e_{cc}\otimes e_{ca}
     }
       \\ \ \\
      \displaystyle{
     +\sum\limits_{\substack{a\neq c \\ c\neq b}}(-1)^{p_c} g^{ac}(z_{13}) g^{cb}(x-y)e_{ac}\otimes e_{bb}\otimes e_{ca}
    +\sum\limits_{\substack{a\neq c \\ b\neq c}}(-1)^{p_c+p_b} g^{ac}(z_{13}) g^{cb}(z_{12})e_{ab}\otimes e_{bc}\otimes e_{ca}.
      }
\end{array}
$$
Finally, the second item on the right-hand side of (\ref{AYBE}) is as follows:
\beq\label{A3}
\begin{array}{c}
\displaystyle{
   R^{y-x}_{23}(z_{23}) R^{x}_{13}(z_{13})=\sum\limits_a f^a(z_{23},y-x)f^{a}(z_{13},x) e_{aa}\otimes e_{aa}\otimes e_{aa}+
   }
      \\ \ \\
      \displaystyle{
   +\sum\limits_{a\neq c} f^a(z_{23},y-x) g^{ca}(x) e_{cc}\otimes e_{aa}\otimes e_{aa}
    +\sum \limits_{a\neq c} (-1)^{p_a} f^a(z_{23},y-x) g^{ca}(z_{13}) e_{ca}\otimes e_{aa}\otimes e_{ac}
     }
\end{array}
\eq
$$
\begin{array}{c}
      \displaystyle{
   +\sum\limits_{a\neq c} g^{ac}(y-x) f^{c}(z_{13},x) e_{cc}\otimes e_{aa}\otimes e_{cc}+\sum\limits_{\substack{a\neq c \\ b\neq c}} g^{ac}(y-x) g^{bc}(x) e_{bb}\otimes e_{aa}\otimes e_{cc}
    }
        \\ \ \\
      \displaystyle{
    +\sum\limits_{\substack{a\neq c \\ b\neq c}}(-1)^{p_c} g^{ac}(y-x) g^{bc}(z_{13}) e_{bc}\otimes e_{aa}\otimes e_{cb}+\sum\limits_{a\neq c}(-1)^{p_c} g^{ac}(z_{23}) f^{a}(z_{13},x) e_{aa}\otimes e_{ac}\otimes e_{ca}
    }
    \\ \ \\
      \displaystyle{
      +\sum\limits_{\substack{a\neq c \\ b\neq c}}(-1)^{p_c} g^{ac}(z_{23}) g^{bc}(x) e_{bb}\otimes e_{ac}\otimes e_{ca}+\sum\limits_{\substack{a\neq c \\ b\neq a}}(-1)^{p_c+p_a} g^{ac}(z_{23}) g^{ba}(z_{13}) e_{ba}\otimes e_{ac}\otimes e_{cb}.
      }
\end{array}
$$
Using the explicit expressions (\ref{A1}), (\ref{A2}) and (\ref{A3}) one can collect the coefficients
 multiplying the same tensor monomials  and prove (\ref{AYBE}) using the following relations:
 \beq\label{A11}
 f^a(z,x)f^a(w,y)-f^a(z+w,y)f^a(z,x-y)-f^a(w,y-x)f^a(z+w,x)=0\,,
 \eq
\beq\label{A12}
g^{ab}(x)g^{bc}(y)-g^{ac}(y)g^{ab}(x-y)-g^{bc}(y-x)g^{ac}(x)=0\,,
\eq
\beq\label{A13}
g^{ac}(x+y)\left(f^a(z,x)-f^a(z,-y)\right)=g^{ac}(x)g^{ac}(y)\,,
\eq
\beq\label{A14}
g^{ac}(z+w)\left(f^c(z,x)+f^c(w,-x)\right)=(-1)^{p_c}g^{ac}(z)g^{ac}(w)\,.
\eq
In this way one gets (\ref{AYBE}).

\paragraph{Proof of Theorem \ref{th2}.}
In the case of the $R$-matrix (\ref{uqsuper}) the associative Yang-Baxter equation does not hold true. In this case we have
\beq
   \displaystyle{
  R^\hbar_{12}(z)
  =\sum\limits_a f^a(z,\hbar) e_{aa}\otimes e_{aa}+
  \tilde{g}(\hbar)\sum\limits_{a\neq c}  e_{aa}\otimes e_{cc}
 +\sum\limits_{a\neq c} {\tilde{g}}^{ac}(z) e_{ac}\otimes e_{ca}\,,
 }
\eq
where
\beq\label{gtilde}
   \displaystyle{
\tilde{g}(\hbar)=\frac{\pi}{\sin(\pi \hbar)}\,, \quad \quad \quad
\tilde{g}^{ac}(z)=\frac{(-1)^{p_b}\pi{\rm sign}(c-a)}{\sin(\pi z)}\,.
}
\eq
For ${\tilde g}^{ac}(z)$ the relations (\ref{A11}), (\ref{A12}), (\ref{A14}) hold. At the
same time, for $\tilde g(\hbar)$ the relation (\ref{A13}) is true, while (\ref{A12}) is replaced with
\beq
   \displaystyle{
\tilde g(x)\tilde g(y)-\tilde g(y)\tilde g(x-y)-\tilde g(y-x)\tilde g(x)=\frac{1}{2\cos\frac{\pi x}{2}\cos\frac{\pi y}{2}\cos\frac{\pi(x-y)}{2}}\,.
}
\eq
This gives  (\ref{relAY}).

\section{Appendix B: normalized $R$-matrices}
\setcounter{equation}{0}
\def\theequation{B.\arabic{equation}}

Here we list the normalized versions of $R$-matrices (\ref{uqsuper})-(\ref{g1,1}). The normalization
(\ref{q04})-(\ref{q05}) gives the $R$-matrices with bars, which are used in the expressions
for commuting difference operators and Hamiltonians of spin chains.

 The normalized version of the $R$-matrix related to ${\hat{\mathcal U}}_q({\rm gl}(N|M))$ (\ref{uqsuper}):
%
  \beq\label{uqsupern}
   \begin{array}{c}
   \displaystyle{
  {\bar\mR}^\hbar_{12}(z)=
   \sum\limits_{a=1}^{N+M} \frac{\sin(\pi(\hbar+(-1)^{p_a}z))}{\sin(\pi(\hbar+z))}\,e_{aa}\otimes e_{aa}
+\frac{\sin(\pi z)}{\sin(\pi (\hbar+z))}\sum\limits_{a\neq b}^{N+M} e_{aa}\otimes e_{bb}+
  }
  \\ \ \\
   \displaystyle{
+
 \frac{\sin(\pi \hbar)}{\sin(\pi (\hbar+z))}\sum\limits_{a< b}^{N+M}
 \Big( (-1)^{p_b} e_{ab}\otimes e_{ba}\,e^{\pi\imath z}+(-1)^{p_a}e_{ba}\otimes
 e_{ab}\,e^{-\pi\imath z}\Big)\,.
  }
  \end{array}
  \eq
  Two simplest examples below are the normalized $R$-matrices (\ref{2,0}) and (\ref{1,1}) respectively.
    \paragraph{Example ${\hat{\mathcal U}}_q({\rm gl}(2|0))$:}
  \beq\label{2,0n}
   \begin{array}{c}
   \displaystyle{
  {\bar\mR}^\hbar_{12}(z)=
   e_{11}\otimes e_{11}+e_{22}\otimes e_{22}+
  }
  \\ \ \\
   \displaystyle{
 +\frac{\sin(\pi z)}{\sin(\pi (\hbar+z))}\Big(e_{11}\otimes e_{22}+e_{22}\otimes e_{11}\Big)+
 \frac{\sin(\pi \hbar)}{\sin(\pi (\hbar+z))}
 \Big(e_{12}\otimes e_{21}e^{\pi\imath z}+e_{21}\otimes e_{12}e^{-\pi\imath z}\Big)\,.
  }
  \end{array}
  \eq
   \paragraph{Example ${\hat{\mathcal U}}_q({\rm gl}(1|1))$:}
\beq\label{1,1n}
   \begin{array}{c}
   \displaystyle{
  {\bar\mR}^\hbar_{12}(z)=
   e_{11}\otimes e_{11}
  +\frac{\sin(\pi (\hbar-z))}{\sin(\pi (\hbar+z))}\,e_{22}\otimes e_{22}+
  }
  \\ \ \\
   \displaystyle{
 +\frac{\sin(\pi z)}{\sin(\pi (\hbar+z))}\Big(e_{11}\otimes e_{22}+e_{22}\otimes e_{11}\Big)+
 \frac{\sin(\pi \hbar)}{\sin(\pi (\hbar+z))}\Big(-e_{12}\otimes e_{21}e^{\pi\imath z}+e_{21}\otimes e_{12}e^{-\pi\imath z}\Big)\,.
  }
  \end{array}
  \eq

The normalized version for the graded extension of the $\mZ_n$-invariant $R$-matrix:
   \begin{equation}\label{superRn}
   \begin{array}{c}
   \displaystyle{
  {\bar R}^\hbar_{12}(z)
  =\sum\limits_{a=1}^{N+M}\frac{\sin(\pi(\hbar+(-1)^{p_a}z))}{\sin(\pi(\hbar+z))}\, e_{aa}\otimes e_{aa}+
   }
  \\ \ \\
     \displaystyle{
  +\sum\limits_{a\neq c}^{N+M} e_{aa}\otimes e_{cc} \exp\Big(\frac{\pi\imath\hbar}{N+M}\Big(2(a-c)-(N+M){\rm sign}(a-c)\Big)\Big)\frac{\sin(\pi z)}{\sin(\pi (\hbar+z))}+
  }
  \\ \ \\
   \displaystyle{
 +\sum\limits_{a\neq c}^{N+M} (-1)^{p_c} e_{ac}\otimes e_{ca} \exp\Big(\frac{\pi\imath z}{N+M}\Big(2(a-c)-(N+M){\rm sign}(a-c)\Big)\Big)\frac{\sin(\pi \hbar)}{\sin(\pi (\hbar+z))}\,.
  }
  \end{array}
  \end{equation}
   The following two examples are the normalized $R$-matrices (\ref{2,0}) and (\ref{1,1}) respectively.
    \paragraph{Example ${\rm GL}(2|0)$:}
  \beq\label{g2,0n}
   \begin{array}{c}
   \displaystyle{
  {\bar R}^\hbar_{12}(z)=
   e_{11}\otimes e_{11}+e_{22}\otimes e_{22}+
  }
  \\ \ \\
   \displaystyle{
 +\frac{\sin(\pi z)}{\sin(\pi (\hbar+z))}\Big(e_{11}\otimes e_{22}+e_{22}\otimes e_{11}\Big)+
 \frac{\sin(\pi \hbar)}{\sin(\pi (\hbar+z))}\Big(e_{12}\otimes e_{21}+e_{21}\otimes e_{12}\Big)\,.
  }
  \end{array}
  \eq
    \paragraph{Example ${\rm GL}(1|1)$:}
  \beq\label{g1,1n}
   \begin{array}{c}
   \displaystyle{
  {\bar R}^\hbar_{12}(z)=
   e_{11}\otimes e_{11}
  +\frac{\sin(\pi (\hbar-z))}{\sin(\pi (\hbar+z))}\,e_{22}\otimes e_{22}+
  }
  \\ \ \\
   \displaystyle{
 +\frac{\sin(\pi z)}{\sin(\pi (\hbar+z))}\Big(e_{11}\otimes e_{22}+e_{22}\otimes e_{11}\Big)+
 \frac{\sin(\pi \hbar)}{\sin(\pi (\hbar+z))}\Big(-e_{12}\otimes e_{21}+e_{21}\otimes e_{12}\Big)\,.
  }
  \end{array}
  \eq
%


\subsection*{Acknowledgments}


We are grateful to A. Liashyk for fruitful discussions and I. Marshall for the careful proofreading.


This work was performed at the Steklov International Mathematical Center and supported by the Ministry of Science and Higher Education of the Russian Federation (agreement no. 075-15-2022-265).


\begin{small}

\end{small}


\begin{thebibliography}{99}
\addcontentsline{toc}{section}{References}

\bibitem{AHZ} A. Antonov, K. Hasegawa, A. Zabrodin,
 {\em On trigonometric intertwining vectors and non-dynamical R-matrix for the Ruijsenaars model},
Nucl. Phys. B503 (1997) 747--770; hep-th/9704074.



\bibitem{ACF} G.E. Arutyunov, L.O. Chekhov, S.A. Frolov,
{\em R-matrix quantization of the elliptic Ruijsenaars–Schneider model}, Commun. Math. Phys. 192 (1998) 405--432.

G.E. Arutyunov, E. Olivucci, {\em Hyperbolic spin Ruijsenaars-Schneider model from Poisson reduction},
Proceedings of the Steklov Institute of Mathematics, 309 (2020) 31--45;
arXiv:1906.02619.

O. Chalykh, M. Fairon,
{\em On the Hamiltonian formulation of the trigonometric spin Ruijsenaars-Schneider system},
	Lett. Math. Phys. 110 (2020), 2893--2940; 	arXiv:1811.08727 [math-ph].

L. Feh\'er, {\em Poisson-Lie analogues of spin Sutherland models}, Nuclear Physics B, 949 (2019) 114807; arXiv:1809.01529 [math-ph].


\bibitem{BMB} B. Basu-Mallick, N. Bondyopadhaya,
{\em Exact partition function of $SU(m|n)$ supersymmetric Haldane–Shastry spin chain},
Nuclear Physics B, 757:3 (2006) 280--302; 	arXiv:hep-th/0607191.

\bibitem{BFG} B. Basu-Mallick, F. Finkel, A. Gonz\'alez-L\'opez,
{\em A novel class of translationally invariant spin chains with long-range interactions},
	J. High Energy Phys. 2020, 99 (2020); arXiv:2003.07217 [cond-mat.stat-mech].

\bibitem{BS} V.V. Bazhanov, A.G. Shadrikov, {\em Trigonometric solutions of triangle equations. Simple Lie superalgebras}, Theoret. and Math. Phys., 73:3 (1987), 1302--1312.

J.H.H. Perk, C.L. Schultz, {\em New families of commuting transfer matrices in q-state
vertex model}, Phys. Lett. A84 (1981) 407.




\bibitem{BDM} O. Blondeau-Fournier, P. Desrosiers, P. Mathieu,
{\em The supersymmetric Ruijsenaars-Schneider model},
Phys. Rev. Lett. 114, 121602 (2015); arXiv:1403.4667 [hep-th].


\bibitem{Chered2} I.V. Cherednik,
  {\em On a method of constructing factorized S matrices in elementary functions},
 Theoret. and Math. Phys., 43:1 (1980) 356–-358.


\bibitem{Drinfeld} V.G. Drinfeld, Quasi-Hopf algebras, Leningrad Math. J., 1:6 (1990), 1419-1457.


\bibitem{FK} S. Fomin, A.N. Kirillov,
{\em Quadratic algebras, Dunkl elements, and Schubert calculus},
Advances in geometry; Prog. in Mathematics book series, 172
 (1999) 147--182.

Anatol N. Kirillov, {\em On Some Quadratic Algebras I $\frac12$:
Combinatorics of Dunkl and Gaudin Elements, Schubert, Grothendieck, Fuss-Catalan, Universal Tutte and Reduced Polynomials}, 	
SIGMA 12 (2016), 002; arXiv:1502.00426 [math.RT].



 \bibitem{GZ} A. Grekov, A. Zotov,
{\em On R-matrix valued Lax pairs for Calogero–Moser models},
J. Phys. A: Math. Theor., 51 (2018), 315202; arXiv: 1801.00245 [math-ph].


A. Grekov, I. Sechin, A. Zotov, {\em Generalized model of interacting integrable tops},
JHEP 10 (2019) 081; arXiv:1905.07820 [math-ph].

E.S. Trunina, A.V. Zotov, {\em Multi-pole extension of the elliptic models of interacting integrable tops}, Theoret. and Math. Phys., 209:1 (2021), 1331--1356; 	arXiv:2104.08982 [math-ph].



\bibitem{HS1}
 F.D.M. Haldane, {\em Exact Jastrow-Gutzwiller resonating-valence-bond ground state of the spin-$\frac{1}{2}$ antiferromagnetic Heisenberg chain with $1/r^2$ exchange}, Phys. Rev. Lett. 60 (1988) 635--638.

B.S. Shastry, {\em Exact solution of an S=1/2 Heisenberg antiferromagnetic
chain with long-ranged interactions}, Phys. Rev. Lett. 60 (1988) 639--642.


\bibitem{Haldane} F.D.M. Haldane, {\em Physics of the Ideal Semion Gas: Spinons and Quantum Symmetries of the Integrable Haldane-Shastry Spin Chain},
Proceedings of the 16th. Taniguchi Symposium on Condensed Matter
Physics, Kashikojima, Japan, October 26-29, 1993, edited by A. Okiji and N. Kawakami
(Springer, Berlin-Heidelberg-New York, 1994); arXiv:cond-mat/9401001.


\bibitem{KrZ}
T. Krasnov, A. Zotov,
{\em Trigonometric integrable tops from solutions of associative Yang-Baxter equation},
 Annales Henri Poincare, 20:8 (2019)
2671--2697;
 arXiv:1812.04209 [math-ph].


\bibitem{KrichZ} I. Krichever, A. Zabrodin,
 {\em Spin generalization of the Ruijsenaars-Schneider model, non-abelian 2D Toda chain and representations of Sklyanin algebra},
Russian Math. Surveys, 50:6 (1995) 1101--1150; arXiv:hep-th/9505039.

\bibitem{KL} S. Krivonos, O. Lechtenfeld,
{\em On N=2 supersymmetric Ruijsenaars-Schneider models},
Physics Letters B, 807, (2020), 135545;
arXiv:2005.06486 [hep-th].


\bibitem{KS} P.P. Kulish, E.K. Sklyanin, {\em Solutions of the Yang-Baxter equation},
Differential geometry, Lie groups and mechanics. Part III, Zap. Nauchn. Sem. LOMI, 95, "Nauka", Leningrad. Otdel., Leningrad, 1980, 129--160; J. Soviet Math., 19:5 (1982), 1596--1620.

W. Galleas, M.J. Martins, {\em R-matrices and Spectrum of Vertex Models based on Superalgebras},
Nucl.Phys. B699: 455-486, 2004; arXiv:nlin/0406003 [nlin.SI].


\bibitem{Lam1} J. Lamers, {\em Resurrecting the partially isotropic Haldane-Shastry model}, 	Phys. Rev. B 97 (2018) 214416; arXiv:1801.05728 [cond-mat.str-el].

\bibitem{Lam2} J. Lamers, V. Pasquier, D. Serban,
 {\em Spin-Ruijsenaars, q-deformed Haldane-Shastry and Macdonald polynomials},
 Commun. Math. Phys. 393, 61–150 (2022);
	arXiv:2004.13210 [math-ph].

\bibitem{Lam3}
J. Lamers, D. Serban,
{\em From fermionic spin-Calogero-Sutherland models to the Haldane-Shastry spin chain by freezing},
	arXiv:2212.01373 [math-ph].


\bibitem{LOZ14} A. Levin, M. Olshanetsky, A. Zotov,
{\em Planck Constant as Spectral Parameter in Integrable Systems and KZB Equations},
JHEP 10 (2014) 109; 	arXiv:1408.6246 [hep-th].

\bibitem{LOZ_rel} A. Levin, M. Olshanetsky, A. Zotov,
{\em Relativistic classical integrable tops and quantum R-matrices},
JHEP 07 (2014) 012; arXiv:1405.7523 [hep-th].

\bibitem{Macd}  I.G. Macdonald, {\em Symmetric functions and Hall polynomials}, Oxford university press, (1998).

\bibitem{MZ1} M. Matushko, A. Zotov,
{\em Anisotropic spin generalization of elliptic Macdonald-Ruijsenaars operators and R-matrix identities},
Ann. Henri Poincar\'e, 24 (2023), 3373--3419;
arXiv:2201.05944 [math.QA].

\bibitem{MZ2} M. Matushko, A. Zotov,
{\em Elliptic generalisation of integrable q-deformed anisotropic Haldane–Shastry long-range spin chain},
Nonlinearity, 36:1 (2023), 319;
	arXiv:2202.01177 [math-ph].

\bibitem{Polych} A.P. Polychronakos, {\em Lattice integrable systems of Haldane-Shastry type},
Phys. Rev. Lett. 70 (1993) 2329--2331.

A.P. Polychronakos, {\em Exact Spectrum of SU(n) Spin Chain with Inverse-Square Exchange}, Nucl. Phys. B419 (1994) 553--566;
	arXiv:hep-th/9310095.

\bibitem{Ruij} S.N.M. Ruijsenaars, {\em Complete integrability of relativistic Calogero-Moser systems and elliptic function identities}, Commun. Math. Phys. 110:2 (1987) 191--213.

\bibitem{SchP}  T. Schedler,
{\em Trigonometric solutions of the associative Yang-Baxter equation},
Mathematical Research Letters, 10:3 (2003) 301–321; arXiv:math/0212258 [math.QA].

A. Polishchuk,
{\em Massey products on cycles of projective lines and trigonometric solutions of the Yang-Baxter equations},
Algebra, Arithmetic, and Geometry, Progress in Mathematics book series, Volume 270,
(2010) 573–617; arXiv:math/0612761 [math.QA].

\bibitem{SeZ} I. Sechin, A. Zotov, {\em R-matrix-valued Lax pairs and long-range spin chains}, Physics Letters B, 781:10 (2018) 1--7;
arXiv:1801.08908 [math-ph].

\bibitem{SeZ2} I.A. Sechin, A.V. Zotov, {\em ${\rm GL}_{NM}$ quantum dynamical R-matrix based on solution of the associative Yang–Baxter equation}, Russian Math. Surveys, 74:4 (2019) 767--769; 	arXiv:1905.08724 [math.QA].

I.A. Sechin, A.V. Zotov, {\em Integrable system of generalized relativistic interacting tops}, Theoret. and Math. Phys., 205:1 (2020) 1292--1303; arXiv: 2011.09599 [math-ph].

E. Trunina, A. Zotov, {\em Lax equations for relativistic GL(NM,C) Gaudin models on elliptic curve},
	2022 J. Phys. A: Math. Theor. 55 395202; arXiv:2204.06137 [nlin.SI].


\bibitem{Uglov} D. Uglov, {\em The trigonometric counterpart of the Haldane Shastry model},
hep-th/9508145.





\end{thebibliography}
\end{document}